\documentclass[submission,copyright,creativecommons]{eptcs}
\pdfoutput=1

\providecommand{\event}{AFL 2017} 
\usepackage{underscore}           

\usepackage{amsfonts}
\usepackage{amssymb}
\usepackage{amsmath}
\usepackage{amsthm}
\usepackage{textcomp}
\usepackage{array}
\usepackage{color}
\usepackage{comment}

\usepackage{graphicx}

\newcommand{\cent}[0]{\mbox{\textcent}}
\newcommand{\dollar}[0]{\$}

\newcommand{\aket}[1]{\langle #1 \rangle}

\newcommand{\mymatrix}[2]{\left( \begin{array}{#1} #2\end{array} \right)}

\newcommand{\myvector}[1]{\mymatrix{c}{#1}}

\newcommand{\mypar}[1]{\left( #1 \right)}
\newcommand{\myceil}[1]{ \left\lceil #1 \right\rceil }

\newcommand{\tildesigma}{\widetilde{\Sigma}}
\newcommand{\tildew}{\tilde{w}}

\newcommand{\pal}{\mathtt{PAL}}

\newcommand{\npal}{\mathtt{NPAL}}

\newcommand{\palnpal}{\mathtt{PAL \mbox{-} NPAL}}
\newcommand{\palnpalyes}{\mathtt{PAL \mbox{-} NPAL_{yes}}}
\newcommand{\palnpalno}{\mathtt{PAL \mbox{-} NPAL_{no}}}

\newcommand{\twin}{\mathtt{TWIN}}
\newcommand{\twint}{{\mathtt{TWIN}(t)}}
\newcommand{\manytwins}{\mathtt{MANYTWINS}}

\newcommand{\ttend}{\mathtt{END}}

\newtheorem{theorem}{Theorem}
\newtheorem{corollary}{Corollary}
\newtheorem{lemma}{Lemma}

\title{Exact Affine Counter Automata\thanks{Parts of the research work were done while Yakary{\i}lmaz was visiting Yamagata University in November 2016 and all authors were visiting Kyoto University in March 2017.}}

\author{Masaki Nakanishi\institute{Department of Education, Art and Science, Yamagata University,\\
		Yamagata, 990--8560, Japan}\email{masaki@cs.e.yamagata-u.ac.jp}
        \and Kamil Khadiev\institute{University of Latvia, Faculty of Computing, Center for Quantum Computer Science, R\={\i}ga, Latvia\\
        Kazan Federal University, Institute of Computational Mathematics and IT,\\ Kremlevskaya str. 18, Kazan, 420008, Russia\\}\email{kamilhadi@gmail.com}
        \and 
        Kri\v{s}j\=anis Pr\=usis \qquad Jevg\=enijs Vihrovs \qquad Abuzer Yakary{\i}lmaz\institute{University of Latvia, Faculty of Computing, Center for Quantum Computer Science,  R\={\i}ga,  Latvia}\email{krisjanis.prusis@lu.lv, jevgenijs.vihrovs@lu.lv, abuzer@lu.lv}       
        }
        
\def\titlerunning{Exact Affine Counter Automata}
\def\authorrunning{M. Nakanishi, K. Khadiev, K. Pr\=usis, J. Vihrovs \& A. Yakary{\i}lmaz} 

\def\copyrightholders{M. Nakanishi, K. Khadiev, K. Pr\=usis, J. Vihrovs \\ A. Yakary{\i}lmaz}

\begin{document}
\maketitle

\begin{abstract}
We introduce an affine generalization of counter automata, and analyze their ability as well as affine finite automata. Our contributions are as follows. We show that there is a language that can be recognized by exact realtime affine counter automata but by neither 1-way deterministic pushdown automata nor realtime deterministic $k$-counter automata. We also show that a certain promise problem, which is conjectured not to be solved by two-way quantum finite automata in polynomial time, can be solved by Las Vegas affine finite automata. Lastly, we show that  how a counter helps for affine finite automata by showing that the language $ \manytwins $, which is conjectured not to be recognized by affine, quantum or classical finite state models in polynomial time, can be recognized by affine counter automata with one-sided bounded-error in realtime.
\end{abstract}

\section{Introduction}

Quantum computation models can be more powerful than their classical counterparts. This is mainly because quantum models are allowed to use negative amplitudes, by which interference can occur between configurations. In order to mimic quantum interference classically, recently a new concept called affine computation was introduced \cite{DCY16A} and its finite automata versions (AfAs) have been examined \cite{DCY16A,VilY16A,BMY16A,HMY17}. Some underlying results are as follows: (i) they are more powerful than their probabilistic and quantum counterparts (PFAs and QFAs) with bounded and unbounded error; (ii) one-sided bounded-error AfAs and nondeterministic QFAs define the same class when using rational number transitions; and, (iii) AfAs can distinguish any given pair of strings by using two states  with zero-error. Very recently, affine OBDD was introduced in \cite{IKPVY17A} and it was shown that they can be exponentially narrower than bounded-error quantum and classical OBDDs.  

In this paper, we introduce (realtime) AfA augmented with a counter (AfCAs), and analyze their ability as well as Las Vegas AfAs. It is already known that AfAs can simulate QFAs exactly by a quadratic increase in the number of states \cite{VilY16A}. However, this simulation cannot be extended to the simulation of QFAs with a counter (QCAs). Therefore, the quantum interference used by QCAs cannot be trivially used by AfCAs. Besides, the well-formed conditions for QCAs can be complicated, but as seen soon, they are easy to check for AfCAs. Thus, we believe that AfCAs may present classical and simpler setups for the tasks done by QCAs.

Our main contribution in this paper is that we show that there is a language that can be recognized exactly (zero-error) by realtime AfCAs but neither by 1-way deterministic pushdown automata nor by realtime deterministic $k$-counter automata. This is the first separation result concerning AfCAs. This is a strong result since an exact one-way probabilistic one-counter automaton (PCA) is simply a one-way deterministic one-counter automaton (DCA) and it is still open whether exact one-way QCAs are more powerful than one-way DCAs and whether bounded-error one-way QCAs are more powerful than one-way bounded-error PCAs (see \cite{SY12A,NY15A} for some affirmative results). 

In \cite{RasY14A}, it was shown that a certain promise problem can be solved by two-way QFAs (2QCFAs) exactly but in exponential time, and bounded-error two-way PFAs (2PFAs) can solve the problem only if they are allowed to use logarithmic amount of memory. We show that the same problem can be solved by realtime Las Vegas AfAs or AfAs with restart in linear expected time. 
Lastly, we address the language $ \manytwins $, which is conjectured not to be recognized by affine, quantum or classical  finite state models in polynomial time. We show how a counter helps for AfAs by showing that $ \manytwins $ can be recognized by AfCAs with one-sided bounded-error in realtime read mode. 

In the next section, we provide the necessary background. Our main results are given in Sections 3, 4 and 5, respectively. Section 6 concludes the paper.

\section{Background}

We assume the reader to have the knowledge of automata theory, and familiarity with the basics of probabilistic and quantum automata. We refer \cite{SayY14A} and \cite{AY15A} for the quantum models. 

Throughout the paper, the input alphabet is denoted as $ \Sigma $ not including the left end-marker ($\cent$) and the right end-marker ($\dollar$). The set $ \tildesigma $ denotes $ \Sigma \cup \{\cent,\dollar\} $. For a given input $ w \in \Sigma^* $, $ |w| $ is the length of $w$, $w[i]$ is the $i$-th symbol of $w$, and $ \tildew = \cent w \dollar $. For any given string $ w \in \{1,2\}^* $, $ e(w) $ denotes the encoding of $ w $ in base-3. The value $ \overline{1} $ in a vector represents the value that makes the vector summation equal to 1, i.e., if the summation of all other entries is $ x $, then it represents the value $ 1-x $.

A (realtime) affine finite automaton (AfA) \cite{DCY16A} $ A $ is a 5-tuple
\[
	A = (S,\Sigma,\{M_\sigma \mid \sigma \in \tildesigma \},s_I,S_a),
\]
where $ S = \{s_1,\ldots,s_n\} $ is a finite set of states, $ \Sigma $ is a finite set of input symbols, $M_\sigma$ is the $ n\times n $ affine transition matrix for symbol $\sigma \in \tildesigma $, $ s_I \in S $ is the initial state, and $ S_a \subseteq S $ is a finite set of accepting states.

We consider a one-to-one correspondence between the set of configurations (i.e., the set of states $ S $) and the standard basis of an $n$-dimensional real vector space. Then, any affine state is represented as an $n$-dimensional real vector such that the summation of all entries is equal to 1. For a given input $w \in \Sigma^*$, $A$ starts its computation in the initial affine state $ v_0 $, where the $I$-th entry is 1 and the others are zeros. Then, it reads $ \tildew $ symbol by symbol from the left to the right and for each symbol the affine state is changed as follows:
\[
	v_j = M_{{\tilde w}[j]} v_{j-1},
\]
where $ 1 \leq j \leq |\tildew| $. To be a well-formed machine, the summation of entries of each column of $ M_\sigma $ must be 1. The final state is denoted as $v_f = v_{|\tildew|}$. At the end, the weighting operator\footnote{This operator returns the weight of each value in the $l_1$ norm of the vector.} returns the probability of observing each state as
\[
	Pr[\mbox{observing } s_i] = \frac{| v_f[i] | }{|v_f|},
\] 
where $ 1 \leq i \leq |S| $, $ v_f[i] $ is the $ i $-th entry of $ v_f $, and $ |v_f| $ is $l_1$ norm of $ v_f $. Thus, the input $w$ is accepted by $ A $ with probability
\[
	f_A(w) = \sum_{s_i \in S_a} \frac{| v_f[i] | }{|v_f|}.
\]

Next, we define a new affine model. 
A (realtime) affine counter automaton (AfCA) $A$ is an AfA augmented with a counter. Formally, it is a 5-tuple
\[
	A = (S,\Sigma,\delta,s_I,S_a),
\]
where the difference from an AfA \[
	\delta: S \times \tildesigma \times \{Z,NZ\} \times S \times \{-1,0,+1\} \rightarrow {\mathbb R} 
\] 
is the transition function governing the behavior of $A$ such that when it is in the state $ s \in S $, reads the symbol $ \sigma \in \tildesigma $, and the current status of the counter is $ \theta \in \{Z,NZ\} $ ($Z:$ zero, $NZ:$ nonzero), it makes the following transition with value $ \delta(s,\sigma,\theta,s',d) $: it switches to the state $ s' \in S $ and updates the value of the counter by $ d \in \{ -1,0,+1 \} $. To be a well-formed affine machine, the transition function must satisfy that, for each triple $ (s,\sigma,\theta) \in S \times \tildesigma \times \{Z, NZ\} $,
\[
	\sum\limits_{s' \in S, d\in \{-1,0,+1\}} \delta(s,\sigma,\theta,s',d) = 1.
\]
Remark that the value of the counter can be updated by a value in $ \{-t,\ldots,+t\} $ for some $ t>1 $ but this does not change the computational power of the model (see \cite{YFSA12A} for more details). 

Any classical configuration of $ A $ is formed by a pair $ (s,c) \in S \times \mathbb{Z} $, where $ s $ is the deterministic state and $ c  $ is the value of the counter. Let $w \in \Sigma^*$ be the given input and $ m = | \cent w \dollar | $. Since all possible values of the counter are in $ \{ -m,\ldots,m \} $, the total number of classical configurations is $ N = m |S| $. We denote the set $ \{ (s,c) \mid s \in S, c \in \{-m,\ldots,m\} \} $ for $ w $ as $\mathcal{C}^w$. In a similar way to AfAs, the automaton $ A $ reads $ \cent w \dollar $ symbol by symbol from the left to the right and $ A $ operates on the classical configurations. Each such configuration, say $ (s,c) $, can be seen as the state of an affine system which we represent as $ \aket{s,c} $ (a  vector in the standard basis of $ \mathbb{R}^N $). During the computation, similarly to quantum models, $ A $ can be in more than one classical configuration with some values, i.e.
\[
 v =\sum_{(s,c) \in \mathcal{C}^w} \alpha_{s,c} \aket{s,c} \mbox{ satisfying that } \sum_{(s,c) \in \mathcal{C}^w} \alpha_{s,c} = 1.
\]
Due to their simplicity, we use such linear combinations to trace the computation of an AfCA. Then we can also define the affine transition matrix $ M_\sigma $ for each $ \sigma \in \tilde{\Sigma} $ as follows:
\[
	M_\sigma\aket{s,c} = \sum_{s'\in S, d\in \{-1, 0, +1\}} \delta(s, \sigma, \theta(c), s', d)\aket{s',c+d}, 
\]
where $ \theta(c) = Z$ if $ c = 0 $, and $ \theta(c) = NZ $ otherwise.
At the beginning of computation, $A$ is in $ v_0 = \aket{s_I,0} $. Then, after reading each symbol, the affine state of the machine is updated, i.e. 
\[
	v_0 \rightarrow v_1 \rightarrow \cdots \rightarrow v_{m}, \mbox{ where } v_{i+1} = M_{\tilde{w}[i+1]}v_i \mbox{ } (0\leq i \leq m-1).
\]
After reading the whole input, the final affine state becomes
\[
	v_f = v_m =  \sum_{(s,c) \in \mathcal{C}^w} \beta_{s,c} \aket{s,c},
\]
and then the weighting operator is applied and the input is accepted with probability
\[
	f_A(w) = \sum_{s \in S_a, c \in \{-m,\ldots, m\}} \frac{|\beta_{s,c}|}{|v_f|},
\] 
which is the total weight of ``accepting'' configurations out of all configurations at the end.

We can extend AfCAs to have multiple counters (affine $k$-counter automata (AfkCAs)) in a straightforward way; the transition function is extended to $ \delta: S\times \tilde{\Sigma} \times \{Z,NZ\}^k \times S \times \{-1,0,+1\}^k \longrightarrow {\mathbb R} $. 

A (realtime) Las Vegas automaton is obtained from a standard one by splitting the set of states into three: the set of accepting, rejecting, and neutral states. When it enters one of them at the end of the computation, then the answers of ``accepting'', ``rejecting'', and ''don't know'' are given, respectively. 

A (realtime) automaton with restart \cite{YS10B} is similar to a Las Vegas automaton, the set of states of it is split into ``accepting'', ``rejecting'', and ``restarting'' states. At the end of the computation, if the automaton enters a restarting state, then all the computation is restarted from the beginning. An automaton with restart can be seen as a restricted sweeping two-way automaton. The overall accepting probability can be simply obtained by making a normalization over the accepting and rejecting probabilities in a single round (see also \cite{YS13A}).  

If an affine automaton is restricted to use only non-negative values as an entry of its transition matrix, then it becomes a probabilistic automaton. As a further restriction, if only 1 and 0 are allowed to be used, then it becomes a deterministic automaton. Thus, any (realtime) AfCA using only 0 and 1 as transition values is a (realtime) deterministic counter automaton (realtime DCA).  

All the models mentioned above are realtime models, whose tape head moves to the right at each step. Next, we introduce a one-way model, whose tape head is allowed to move to the right or stay at the same position, but not allowed to move to the left.

A one-way deterministic pushdown automaton (1DPA) $ A $ is a 7-tuple
\[ A=(S, \Sigma, \Gamma, \delta, s_I, Z_0, S_a), \]
where $ S=\{s_1,\ldots, s_n\} $ is a finite set of states, $ \Sigma $ is a finite set of input symbols, $ \Gamma $ is a finite set of stack symbols, $ \delta : S\times \tilde{\Sigma} \times \Gamma \times S \times \Gamma^* \times \{0,1\} \longrightarrow \{0,1\} $ is a transition function, $ s_I $ is the initial state, $ Z_0 $ is the initial stack symbol, and $ S_a\subseteq S $ is the set of accepting states. To be a well-formed machine, the transition function must satisfy that, for each triple $ (s, \sigma, \gamma) $, 
\[
\sum_{s'\in S, \gamma'\in \Gamma^*, D\in\{0,1\}} \delta(s, \sigma, \gamma, s', \gamma', D) = 1.
\] 

For a given input $ \cent w \dollar $, the automaton $ A $ starts its computation with the following initial configuration:
\begin{itemize}
	\item the initial state is $ s_I $,
	\item the stack has only the initial stack symbol $ Z_0 $,
	\item the tape head points to the left endmarker.
\end{itemize}
Then, at each step of the computation, $ A $ is updated according to the transition function $ \delta $, i.e., $ \delta(s, \sigma, \allowbreak \gamma,  s', \bar{\gamma'}, D) =1 $ implies that if the current state is $ s $, the scanned input symbol is $ \sigma $ and the stack-top symbol is $ \gamma $, then it moves to the state $ s' $ and updates the stack by deleting the stack-top symbol and pushing $ \bar{\gamma'} $. Also, the tape head moves to $D$ where $ D=0 $ means ``stationary'' and $ D=1 $ means ``move to the right''. Note that a move with $ D=0 $ is called an $\varepsilon$-move.
For an input word $ w $, if $ A $ reaches an accepting state, then $ w $ is accepted.


A promise problem ${\tt P} \subseteq \Sigma^* $ is formed by two disjoint subsets: yes-instances $ \tt P_{yes} $ and no-instances $ \tt P_{no} $. An automaton $ A $ solves $ P $ with error bound $ \epsilon < \frac{1}{2} $ if each yes-instance (resp. no-instance) is accepted (resp. rejected) with probability at least $ 1 - \epsilon $. If all yes-instances (resp. no-instances) are accepted with probability 1 (resp. 0), then the error bound is called one-sided. If $ \epsilon = 0 $, then the problem is said to be solved exactly (or with zero-error). A promise problem is solved by a Las Vegas algorithm with success probability $ p<1 $ if any yes-instance (resp. no-instance) is accepted (resp. rejected) with probability $ p' \geq p $ and the answer of ``don't know'' is given with the remaining probability $ 1 - p' $. If $ \mathtt{P_{yes}} \cup \mathtt{P_{no}} = \Sigma^* $, then it is called language recognition (for $ \tt P_{yes} $) instead of solving a promise problem.

\section{Exact separation}

We start with defining a new language $ \ttend $:
\[
 \ttend = \{ w \in \{0,1,2\}^*2\{0,1,2\}^* \mid w^r[{|w|_2}] = 1 \},
\]
where $ |w|_2 $ is the number of symbols $2$ in $ w $, $ w^r $ is the reverse of $w$, and $ w^r[{|w|_2}] $ is the $ (|w|_2) $-th symbol of $ w^r $.

\begin{theorem}
	The language $ \ttend $ is recognized by an AfCA $ A $ exactly.
\end{theorem}

\begin{proof}
%
We will use two states $ (s_1,s_2) $ for deterministic computation and five states ($ p_0,p_1,p_2,p_3,p_4 $) for affine computation. The initial states are $ s_1 $ and $p_0 $. In other words, we consider a product of a 2-state deterministic finite automaton and a 5-state affine counter automaton.

The classical part is responsible for checking whether the given input has at least one symbol 2. For this purpose, $ s_1 $ switches to $ s_2 $ after reading a symbol $2 $ and then never leaves $ s_2 $ until the end of the computation. If the automaton ends in state $ s_1 $, the input is rejected.

From now on, we focus on the affine transitions. 
The computation starts in the following affine configuration:
\[
	\aket{p_0,0},
\]
where $ p_0 $ is the affine state and $ 0 $ represents the counter value. 

After reading the left end-marker, $ p_0 $ goes to $ p_0 $, $ p_1 $, and $ p_2 $ with the values  $1$, $1 $, and $ -1 $, respectively, without changing the counter value. Then, the affine state becomes
\[
	\aket{p_0,0}+\aket{p_1,0}-\aket{p_2,0}.
\]

We list the all transitions until reading the right end-marker below, in which $ c $ can be any integer representing the counter value. Remark that the counter status is never checked in these transitions.
\begin{itemize}
	\item When reading a symbol 0:
		\begin{itemize}
			\item $ \aket{p_0,c} \rightarrow \aket{p_0,c} $
			\item $ \aket{p_1,c} \rightarrow \aket{p_1,c} - \frac{1}{2} \aket{p_3,c+1} + \frac{1}{2} \aket{p_4,c+1} $
			\item $ \aket{p_2,c} \rightarrow \aket{p_2,c} $
			\item $ \aket{p_3,c} \rightarrow \aket{p_3,c+1} $
			\item $ \aket{p_4,c} \rightarrow \aket{p_4,c+1} $
		\end{itemize}
	\item When reading a symbol 1:
		\begin{itemize}
			\item $ \aket{p_0,c} \rightarrow \aket{p_0,c} $
			\item $ \aket{p_1,c} \rightarrow \aket{p_1,c} + \frac{1}{2} \aket{p_3,c+1} - \frac{1}{2} \aket{p_4,c+1} $
			\item $ \aket{p_2,c} \rightarrow \aket{p_2,c} $
			\item $ \aket{p_3,c} \rightarrow \aket{p_3,c+1} $
			\item $ \aket{p_4,c} \rightarrow \aket{p_4,c+1} $
		\end{itemize}
	\item When reading a symbol 2:
		\begin{itemize}
			\item $ \aket{p_0,c} \rightarrow \aket{p_0,c} $
			\item $ \aket{p_1,c} \rightarrow \aket{p_1,c-1} - \frac{1}{2} \aket{p_3,c} + \frac{1}{2} \aket{p_4,c} $
			\item $ \aket{p_2,c} \rightarrow \aket{p_2,c-1} $
			\item $ \aket{p_3,c} \rightarrow \aket{p_3,c} $
			\item $ \aket{p_4,c} \rightarrow \aket{p_4,c} $
		\end{itemize}
\end{itemize} 
Let $ w $ be the input, let $ n=|w| $, $ x = w^r $, and $ |w|_2 = k \geq 1 $. The affine state before reading the right end-marker is
\[
	\aket{p_0,0} + \aket{p_1,-k} - \aket{p_2,-k} + \sum_{i=1}^{|x|=n} (-1)^{x[i]} \mypar{ -\frac{1}{2} \aket{p_3,i-k} + \frac{1}{2} \aket{p_4,i-k} }.
\] 
Here the first three terms are trivial since (i) $ \aket{p_0,0} $ never leaves itself, and, (ii) the values of $ p_1 $ and $ p_2 $ are never changed and the counter value is   decreased for each symbol 2 ($ k $ times in total). 

In order to verify the last term, we closely look into the step when reading an arbitrary input symbol, say $ w[j] $ ($ 1 \leq j \leq n $). Suppose that $ t \in \{0,\ldots,j-1 \} $ symbols 2 have been read until now. We calculate the final counter values of the configurations with states $ p_3 $ and $ p_4 $ that are created from $ \aket{p_1,-t} $ in this step. 
\begin{itemize}
\renewcommand\labelitemi{$\bullet$}
\item If $ w[j] $ is $ 0 $ or $ 1 $, then the following configurations are created:
\[
	-\frac{1}{2} \aket{p_3,-t+1} + \frac{1}{2} \aket{p_4,-t+1} \mbox{ or } \frac{1}{2} \aket{p_3,-t+1} - \frac{1}{2} \aket{p_4,-t+1},
\]
respectively. In the remaining part of the computation, $ k-t $ symbols 2 and $ n-j-(k-t) $ symbols 0 or 1 are read. When reading a symbol 2, the counter value remains the same and it is increased by 1 when a symbol 0 or 1 is read. Thus, their final counter values hit $ -t+1+n-j-(k-t) = (n-j+1)-k $.
\item If $ w[j] $ is 2, then the following configuration is created:
\[
	-\frac{1}{2} \aket{p_3,-t} + \frac{1}{2} \aket{p_4,-t}.
\]
In the remaining part of the computation, $ k-t-1 $ symbols 2 and $ n-j-(k-t-1) $ symbols 0 or 1 are read. Then, as explained in the previous item, their final counter values hit $ -t+n-j-(k-t-1) = (n-j+1)-k $.
\end{itemize}
It is clear that $ n-j+1 $ refers to $ i $ in the above equation, i.e. $ x[i] = x[n-j+1] = w[j] $, and so the correctness of this equation is verified. Moreover, we can follow that the counter value for these configurations is zero if and only if $ i = n-j+1 = k $, and, this refers to the input symbol $ x[|w|_2] = (w^r)[|w|_2] $. 

Therefore, if we can determine the value of $ w^r[|w|_2] $ with zero error, then we can also determine whether the given input is in the language or not with zero error. For this purpose, we use the following transitions on the right end-marker in which the value of the counter is not changed:
\begin{itemize}
\item Both states $ p_1 $ and $ p_2 $ switch to $ p_1 $. Thus,  the pair $  \aket{p_1,-k} -  \aket{p_2,-k} $ disappears.
\item If the value of the counter is non-zero, both states $ p_3 $ and $ p_4 $ switch to $ p_3 $. Thus, each pair of the form $ (-1)^{x[i]} \mypar{ \frac{1}{2} \aket{p_1,\neq 0} - \frac{1}{2} \aket{p_2,\neq 0} } $ disappears ($ i \neq k $).
\item If the value of the counter is zero, both states $ p_3 $ and $ p_4 $ switch to themselves. Moreover, the state $ p_0 $ switches to $ p_3 $ and $ p_4 $ with values of $ \frac{1}{2} $. Then, the following interference appears:
\[
	v_f =
	(-1)^{x[i]} \mypar{ -\frac{1}{2} \aket{p_3,0} + \frac{1}{2} \aket{p_4,0} }
	+
	\frac{1}{2} \aket{p_3,0} + \frac{1}{2} \aket{p_4,0}.
\]
\end{itemize}
If $ x[i] = 1 $, then $ v_f = \aket{p_3,0} $. If $ x[i] = 0 $ or $ x[i] =2 $, then $ v_f = \aket{p_4,0} $. Thus, by setting $ (s_2,p_3) $ as the only accepting state, we can obtain the desired machine.
 \end{proof}

Next, we prove that the language $ \ttend $ is recognized neither by 1-way deterministic pushdown automata (1DPAs) nor by realtime deterministic $k$-counter automata (realtime DkCAs)\footnote{Since 1-way (with epsilon moves) deterministic 2-counter automata can simulate Turing machines, the restriction of ``realtime'' is essential.}. For this purpose, we introduce the following lemma (the pumping lemma for deterministic context-free languages (DCFLs)) \cite{Yu89}.
\begin{lemma}
	\label{lemma:pumping}
	(Pumping Lemma for DCFLs \cite{Yu89})
	Let $ L $ be a DCFL. Then there exists a constant $ C $ for $ L $ such that for any pair of words $ w $, $ w' \in L $ if
	\begin{enumerate}
		\renewcommand{\labelenumi}{(\arabic{enumi})}
		\item $ w = xy $ and $ w' = xz, |x| > C $, and
		\item $ ^{(1)}y = \,^{(1)}z $, where $ ^{(1)}w $ is defined to be the first symbol of $ w $
		\[
		^{(1)}w = \left \{ 
		\begin{array}{ll}
		x & \mbox{ if } |w| > 1, w=xy, \mbox{ and } |x| = 1;\\
		w & \mbox{ if } |w| \leq 1,
		\end{array}\right .
		\]
	\end{enumerate}
	then either (3) or (4) is true:
	\begin{enumerate}\setcounter{enumi}{2}
		\renewcommand{\labelenumi}{(\arabic{enumi})}
		\item there is a factorization $ x=x_1 x_2 x_3 x_4 x_5, |x_2 x_4| \geq 1 $ and $ |x_2 x_3 x_4| \leq C $, such that for all $ i \geq 0 $ $ x_1 x_2^i x_3 x_4^i x_5 y $ and $ x_1 x_2^i x_3 x_4^i x_5 z $ are in $ L $;
		\item there exist factorizations $ x = x_1 x_2 x_3 $, $ y = y_1 y_2 y_3 $ and $ z = z_1 z_2 z_3 $, $ |x_2| \geq 1 $ and $ |x_2 x_3| \leq C $, such that for all $ i \geq 0 $ $ x_1 x_2^i x_3 y_1 y_2^i y_3 $ and $ x_1 x_2^i x_3 z_1 z_2^i z_3 $ are in $ L $. \qed
	\end{enumerate}
\end{lemma}

\begin{theorem}
	\label{theorem:impossibilityDPA}
	The language $ \ttend $ cannot be recognized by any 1DPA.
\end{theorem}

\begin{proof}
	We assume that $ \ttend $ is a DCFL and let $ C $ be the constant for $ \ttend $ in Lemma~\ref{lemma:pumping}. Choose $ w = 2^p 1 0^{p-1} \in \tt END $ and $ w' = 2^p 1 0^{p-1} 1 0^{p-1} $ for some integer $ p>C+1 $, and set $ x = 2^p 1 0^{p-2}, y = 0 $ and $ z = 0 1 0^{p-1} $. Then, $ w = xy $ and $ w' = xz $ satisfy (1) and (2) in Lemma~\ref{lemma:pumping}.
	
	We first consider the case that (3) holds. In order to satisfy $ x_1 x_2^i x_3 x_4^i x_5 y \in \ttend $ for $ i\geq 0 $, $ x_2 x_4 $ must not have the symbol 1 (otherwise, $ x_1 x_2^0 x_3 x_4^0 x_5 y \not\in \ttend $). Thus, $ x_2 $ and $ x_4 $ are of the form $ 2^t $ or $ 0^t $ for some constant $ t $. If $ x_2 = 2^{t_1} $ and $ x_4 = 2^{t_2} $,  $ x_1 x_2^i x_3 x_4^i x_5 y \not\in \ttend $ for $ i\neq 1$. Similarly, $ x_2 = 0^{t_1} $ and $ x_4 = 0^{t_2} $ cannot occur. Thus, the only possible choice is $ x_2 = 2^{t_1} $ and $ x_4 = 0^{t_2} $ for some $ t_1 $ and $ t_2 $. In order to satisfy $x_1 x_2^i x_3 x_4^i x_5 y \in \ttend $ for $ i\geq 0 $, $ t_1 = t_2$ must hold. However, This causes $ x_1 x_2^i x_3 x_4^i x_5 z \not\in \ttend $ for $ i\neq 1 $. Thus, (3) does not hold.
	
	Next, we consider the case that (4) holds. Since $ |x_2 x_3 | \leq C $, $ x_2 $ can have only 0s. Thus, for any factorization $ w = x_1 x_2 x_3 y_1 y_2 y_3 $, $ x_1 x_2^i x_3 y_1 y_2^i y_3 \not\in \ttend $ for $ i \neq 1 $. Thus, (4) does not hold. This is a contradiction. Therefore, $ \ttend $ is not a DCFL, which implies no 1DPA can recognize $ \ttend $. \end{proof}


\begin{theorem}
	\label{theorem:impossibilityDkCA}
	The language $ \ttend $ cannot be recognized by any realtime DkCA.
\end{theorem}

\begin{proof}
	We assume that there exists a realtime DkCA that recognizes $ \ttend $. We consider an input of the form $ w = 2^m x y (x\in \{0,1\}^m, y\in \{0,1\}^*) $. Then we have $ 2^m $ possible $ x $'s. For any $ x_1 $ and $ x_2 \in \{0,1\}^m (x_1 \neq x_2) $, we will show that there exists a $ y $ such that $ 2^m x_1 y \in \ttend $ and $ 2^m x_2 y \not\in \ttend $ or vice versa.
    
	We assume that $ x_1[i] \neq x_2[i] $. Note that there exists such an $ i $ since $ x_1 \neq x_2 $. We also assume that $ x_1[i] = 1 $ and $ x_2[i] = 0 $ without loss of generality. We set $ y=0^{i-1} $. Then $ (2^m x_1 y)^R[m] = x_1[i] = 1 $ and $ (2^m x_2 y)^R[m] = x_2[i] = 0 $. Thus, $ 2^m x_1 y \in \ttend $ and $ 2^m x_2 y \not\in \ttend $. Therefore, the configurations after reading $ 2^m x_1 $ and $ 2^m x_2 $ must be different. However, the number of possible configurations for a realtime DkCA after reading the partial input $ 2^m x $ is $ O(m) $ while there are $ 2^m $ possible $ x $'s. This is a contradiction.
	\end{proof}

Currently, we do not know any QCA algorithm solving $ \ttend $. Moreover, recently another promise problem solvable by exact QCAs but not by DCAs was introduced in \cite{NY15A} and we also do not know whether AfCAs can solve this promise problem.

\section{Las Vegas algorithms}

In \cite{RasY14A}, some promise problems were given in order to show the superiority of two-way QFAs (2QCFAs) over two-way PFAs (2PFAs). We show that the same problem can be solved by realtime Las Vegas AfAs or AfAs with restart in linear expected time. 

First we review the results given in \cite{RasY14A}. Let $ \pal = \{ w \in \{1,2\}^* \mid w = w^r \} $ be the language of palindromes. Based on $ \pal $, the following promise problem is defined: $ \palnpal $ composed of
\begin{itemize}
\item $ \palnpalyes = \{ x 0 y \mid x \in \pal, y \not\in \pal \} $ and
\item $ \palnpalno \mspace{7.1mu} = \{ x 0 y \mid x \not\in \pal, y \in \pal \} $.
\end{itemize}

It was shown that $ \palnpal $ can be recognized by 2QCFAs exactly but in exponential time and bounded-error 2PFAs can recognize $ \palnpal $ only if they are allowed to use a logarithmic amount of memory. Now we show that $ \palnpal $ can be recognized by realtime Las Vegas AfAs and so also by AfAs with restart in linear expected time.

\begin{theorem}
	The promise problem $ \palnpal $ can be solved by Las Vegas AfA $ A $ with any success probability $ p < 1 $.
\end{theorem}

\begin{proof}
	It is known that AfAs can recognize  $ \pal $ with one-sided bounded-error \cite{VilY16A,YS10B} and so we can design a Las Vegas automaton for $ \palnpal $ by using similar ideas given in \cite{RasY14A,GefY15A}.
	
	The automaton $ A $ has 5 states $ S = \{ s_1,\ldots,s_5 \} $ where $ s_1 $ and $ s_2 $ are accepting states; $ s_3 $ and $s_4$ are rejecting states; and $s_5$ is the only neutral state. After reading $ \cent $, the affine state is set to $ v_1 = (0~~0~~1~~0~~0)^T $. Remember that $ e(u) $ denotes the encoding of the string $ u \in \{1,2\}^* $ in base-3.
	
	We apply the following operators when reading symbols 1 and 2:
	\[
	M_1 = \mymatrix{rrrrr}{4 & 1 & 1 & 1 & 1 \\ 0 & 1 & 1 & 0 & 0 \\ 0 & 0 & 3 & 0 & 0 \\ 0 & 0 & 0 & 1 & 0 \\ -3 & ~ -1 & ~ -4 & ~ -1 & ~~~ 0 }
	\mbox{ and }
	M_2 = \mymatrix{rrrrr}{5 & 2 & 2 & 2 & 2 \\ 0 & 1 & 2 & 0 & 0 \\ 0 & 0 & 3 & 0 & 0 \\ 0 & 0 & 0 & 1 & 0 \\  -4 & ~ -2 & ~ -6 & ~ -2 & ~-1 }
	\] 
	that encode strings $ u $ and $ u^r $ into the values of the first and second states in base-3 after reading $u \in \{1,2\}^*$. Here the third entry helps for encoding $ u^r $, the fourth entry is irrelevant to encoding, and the fifth entry is used to make the state a well-defined affine vector. By using induction, we can show that $ M_1 $ and $ M_2 $ do the aforementioned encoding if the first three entries are respectively 0, 0, and 1.  
	
	For $ u = 1 $ or $ u = 2 $, we can have respectively
	\[
	v_{|\cent 1|} = \mymatrix{r}{1 \\ 1 \\ 3 \\ * \\ \overline{1}}
	\mbox{ and }
	v_{|\cent 2|} = \mymatrix{r}{2 \\ 2 \\ 3 \\ * \\ \overline{1}}.
	\] 
	Suppose that $ u $ is read, then we have the following affine state
	\[
	v_{|\cent u|} = \mymatrix{c}{ e(u) \\ e(u^r) \\ 3^{|u|} \\ * \\ \overline{1}}.
	\]
	By using this, we can calculate the new affine states after reading $ u1 $ and $ u2 $ as
	
	{
		\[
		v_{|\cent u1|} = \myvector{ 3e(u)+1 = e(u1) \\ e(u^r) + 3^{|u|} = e(1u^r) \\ 3^{|u1|} \\ * \\ \overline{1} }
		\mbox{ and }
		v_{|\cent u2|} = \myvector{ 3e(u)+2 = e(u2) \\ e(u^r) + 2 \cdot 3^{|u|} = e(2u^r) \\ 3^{|u2|} \\ * \\ \overline{1}},
		\]
	}
	respectively. Thus, our encoding works fine.
	
	Let $ x0y $ be the input as promised. Then, before reading the symbol 0, the affine state will be
	\[
	v_{|\cent x|} = \myvector{ e(x) \\ e(x^r) \\ 3^{|x|} \\ 0 \\ \overline{1} }. 
	\]
	For the symbol 0, we apply the following operator:
	\[
	M_0 = \mymatrix{rrrrr}{ 0 & 0 & 0 & 0 & 0 \\ 0 & 0 & 0 & 0 & 0 \\ 1 & 1 & 1 & 1 & 1 \\ 1 & ~ -1 & 0 & 0 & 0 \\ -1 & 1 & ~~~ 0 & ~~~ 0 & ~~~ 0 }
	\]
	After reading 0, the new affine state will be
	\[
	v_{|\cent x0|} = \myvector{ 0 \\ 0 \\ 1 \\ e(x) - e(x^r) \\ \overline{1} },
	\]
	where the first three entries are set to 0, 0, and 1 for encoding $ y $, and, the difference $ e(x) - e(x^r) $ is stored into the fourth entry. 
	
	Similarly to above, after reading $ y $, the affine state will be
	\[
	v_{|\cent x0y|} = \myvector{ e(y) \\ e(y^r) \\ 3^{|y|} \\ e(x) - e(x^r) \\ \overline{1} }.
	\]
	Then, the end-marker is read before the weighting operator is applied. Let $ k $ be an integer parameter. The affine operator for the symbol $ \dollar $ is 
	\[
	M_\dollar(k) = \mymatrix{rrrrr}{k & ~ - k & ~~~0 & 0 & 0 \\  -k & k & 0 & 0 & 0 \\ 0 & 0 & 0 & k & 0 \\ 0 & 0 & 0 & ~ -k & 0 \\ 1 & 1 & 1 & 1 & ~~~ 1 }
	\] 
	and so the final state will be
	\[
	v_f = \myvector{ k (e(y) - e(y^r)) \\ -k( e(y) - e(y^r) ) \\ k (e(x) - e(x^r)) \\ -k (e(x) - e(x^r)) \\ 1 }.
	\]
	
	If the input $ x0y $ is a yes-instance, then $ x \in \pal $ and $ y \in \npal $. Thus, $ e(x)-e(x^r) $ is zero and $ |e(y)-e(y^r)| $ is at least 1. In such a case, after the weighting operator, the input is accepted with probability at least $ \frac{2k}{2k+1} $  and the answer of ``don't know'' is given with probability at most $ \frac{1}{2k+1} $.
	
	If the input $ x0y $ is a no-instance, then $ x \in \npal $ and $ y \in \pal $. Thus, $ |e(x)-e(x^r)| $ is at least 1 and $ e(y)-e(y^r) $ is zero. In such a case, after the weighting operator, the input is rejected with probability at least $ \frac{2k}{2k+1} $  and the answer of ``don't know'' is given with probability at most $ \frac{1}{2k+1} $.
	By picking a sufficiency big $ k $, the success probability can be arbitrarily close to 1.
	\end{proof}

\begin{corollary}
	The promise problem $ \palnpal $ can be solved by an exact AfA with restart in linear expected time.
\end{corollary}
\begin{proof}
	In the above proof, we change the neutral states to restarting states, and then obtain the desired machine. For any promised input, the input is either only accepted or only rejected. Since the success probability is constant ($ p $), the expected runtime is $ \frac{1}{p} |w| $ for the promised input $w$.
\end{proof} 

We conjecture that bounded-error 2QCFAs cannot solve $ \palnpal $ and $ \pal $ in polynomial time. Moreover, we leave open whether there exists a promise problem (or a language) solvable by bounded-error AfAs but not by 2QCFAs. 

\section{Bounded-error algorithms}

Similar to $ \pal $, the language $ \twin =  \{ w0w \mid w \in \{1,2\}^* \} $ can be also recognized by one-sided bounded-error AfAs (see also \cite{VilY16A}). After making some straightforward modifications, we can show that the language
\[
 \twint = \{ w_1 0 w_2 0 \cdots 0 w_t 3 w_t 0 \cdots 0 w_2 0 w_1 \mid w_i \in \{1,2\}^*, 1 \leq i \leq t  \}
\]
for some $ t>0 $ can also be recognized by negative one-sided bounded-error AfAs.

Since it is a non-regular language, it cannot be recognized by bounded-error PFAs and QFAs \cite{AY15A}. On the other hand, we can easily give a bounded-error 2QCFA algorithm for $\twint$ \cite{YS10B} but similarly to $ \pal $ it runs in exponential expected time. By using the impossibility proof given for $\pal$ \cite{DS92,RasY14A}, we can also show that $ \twint $ can be recognized by 2PFAs only if augmented with a logarithmic amount of memory. From the literature \cite{Ros66},\footnote{In the original language, there is a symbol 0 instead of the symbol 3. But since $t$ is fixed, the middle 0 can be easily detected by using internal states and so the results regarding the original language still hold for this modified version.} we also know that this language can be recognized by a DFA having at least $ k $ heads, where
\[
	t \leq \myvector{k \\ 2} 
	\mbox{ or }
	k = \myceil{\sqrt{2t+\frac{1}{4}}-\frac{1}{2}} .
\] 
Moreover, using nondeterminism and additional pushdown store does not help to save a single head \cite{ChLi88}. Bounded-error PFAs can recognize $ \twint $ by using two heads but the error increases when $ t $ gets bigger \cite{Yak11B,Yak12B}. For a fixed error, we do not know any PFA algorithm using a fixed number of heads. The same result is also followed for bounded-error QFAs with a stack. (It is open whether bounded-error PFAs can recognize $ \twint $ by using a stack \cite{YFSA12A}.) 

Based on $ \twint $, we define a seemingly harder language $ \manytwins $ that is defined by the union of all $ \twint $s:
\[
	\manytwins = \bigcup_{t = 1}^\infty \twint .
\]
Since the number $t$ is not known in advance, we do not know how to design a similar algorithm for the affine, quantum, and classical models discussed above. On the other hand, this language seems a good representative example for how a counter helps for AfAs.

\begin{theorem}
	The language $ \manytwins $ can be recognized by an AfCA $A$ with one-sided bounded-error arbitrarily close to zero.
\end{theorem}

\begin{proof}
	The automaton $A$ has 10 states: $ \{s_1,s_2,s_3,s'_1,s'_2,s'_3,s_e,s'_e,s_a,s_r\} $, $ s_1 $ is the initial state, and $s_a$ is the only accepting state.
	
	Let $ k $ be an arbitrarily big integer. If there is no symbol 3, then the automaton $A$ never switches to the state $ s_a $ and so the input is accepted with zero probability. We assume then the input has at least one symbol 3 from now on.
	
	The automaton $A$ stays in $s_1$ without changing the value of the counter when reading $ \cent $. Then, until reading the first $ 3 $, it uses the following transitions. 
	
	 Let $ u_1 = w_1 0 w_2 0 \cdots 0 w_t 3 $ be the prefix of the input until the first $ 3 $, where $ w_i \in \{1,2\}^* $ for each $ i \in \{1,\ldots,t\} $ and $ t \geq 1 $.
	 
	When reading a block of $ \{1,2\}^* $, say $w_i$, before a symbol 0 or the symbol 3, it encodes $ w_i $ into the value of $ s_2 $ in base-3 by help of the states $s_1$ and $s_3$. If $ w_i $ is the empty string, then the value of $ s_2 $ becomes 0. During encoding, the value of $ s_1 $, which is 1,  does not change and the value of $ s_3 $ is updated to have a well-formed affine state.
	
	 After reading a 0:
	\begin{itemize}
	\item It stays in $ s_1 $ and increases the value of the counter by 1.
	\item The value of $ s_2 $ is $ e(w_i) $ before the transition. Then the values of $ s_e $ and $ s'_e $ are set to $ ke(w_i) $ and $-ke(w_i)$, respectively, and the value of the counter does not change. Moreover, the value of $s_2$ is set to zero. 
	\item Due to the above transitions, the value of $ s_3 $ is automatically set to zero.
	\end{itemize}
	
	After reading the first $ 3 $:
	\begin{itemize}
	\item It switches from $ s_1 $ to $ s'_1 $ without changing the value of the counter.
	\item The value of $ s_2 $ is $ e(w_t) $ before the transition. Then the values of $ s_e $ and $ s'_e $ are set to $ ke(w_t) $ and $-ke(w_t)$, respectively, and the value of the counter does not change. Moreover, the value of $s_2$ is set to zero. 
	\item Due to the above transitions, the value of $ s_3 $ is automatically set to zero.
	\end{itemize}	
	 Then, after reading $u_1$, the affine state will be
	 \[
	 	\aket{s'_1,t-1} + \sum\limits_{i=1}^t \left( ke(w_i)\aket{s_e,i-1} - ke(w_i)\aket{s'_e,i-1} \right),
	 \]
	 where, by using the different values of the counter, $k$ times the encoding of each $ w_i $ is stored as the values of $ s_e $ and $ s'_e $. If $ t=0$ ($ u_1 = 3 $), then the affine state will be $ \aket{s'_1,0} $. 
	 
	 If after reading $u_1$ the automaton reads another symbol $ 3 $, then it switches to $ s_r $ from $ s'_1 $, $ s'_2 $, and $ s'_3 $, and then stays there until the end of the computation. Thus, in such a case, the input is also accepted with zero probability. Therefore, in the last part, we assume that the input does not have another symbol $3$. 
	 
	 Let $ u_2 = w'_z 0 w'_{z-1} 0 \cdots 0 w'_1 \dollar $ be the part to be read after the symbol 3, where $ w'_j \in \{1,2\}^* $ for each $ j \in \{1,\ldots,z\} $ and $ z>0 $. With a similar strategy, when reading a block of $ \{1,2\}^* $, say $ w'_j $, before a symbol 0 or the symbol $\dollar$, the automaton encodes it into the value of $ s'_2 $ in base 3 by the help of states $ s'_1 $ and $ s'_3 $. If $ w'_j $ is the empty string, then the value of $ s'_2 $ is zero. During the encoding, the value of $ s'_1 $, which is 1, does not change and the value of $ s'_3 $ is updated to have a well-formed affine state.

 After reading a 0:
	\begin{itemize}
	\item It stays in $ s'_1 $ and decreases the value of the counter by 1.
	\item The value of $ s'_2 $ is $ e(w'_j) $ before the transition. Then the values of$ -ke(w'_j) $ and $ke(w'_j)$ are added to $ s_e $ and $ s'_e $, respectively, and the value of the counter does not change. Moreover, the value of $s'_2$ is set to zero. 
	\item Due to the above transitions, the value of $ s'_3 $ is automatically set to zero.
	\end{itemize}
	
	After reading the $ \dollar $:
	\begin{itemize}
	\item It switches from $ s'_1 $ to $ s_a $ without changing the value of the counter.
	\item The value of $ s'_2 $ is $ e(w'_1) $ before the transition. Then the values of $ -ke(w'_1) $ and $ke(w'_1)$ are added to $ s_e $ and $ s'_e $, respectively, and the value of the counter does not change. Moreover, the value of $s'_2$ is set to zero.
	\item Due to the above transitions, the value of $ s'_3 $ is automatically set to zero.
	\end{itemize}	
	
	If $ z = 0 $, then the only transition is switching from $ s'_1 $ to $ s_a $.
	
	Suppose that $ t = z > 0 $. Then it is clear that if $ w_t = w'_z $, then the values of the affine state $ \aket{s_e,t-1} $ and $ \aket{s'_e,t-1} $ will be set to zero. Otherwise, their values will respectively be $ k(e(w_t)-e(w'_z)) $ and $ -k(e(w_t)-e(w'_z)) $, the absolute value of each will be at least $ k $. The same situation holds for each pair $ (w_i,w'_j) $ where $ 1 \leq i=j \leq t $. That means, if the input is a member (including the case of $ t=z=0 $), then the final affine state will be
	$
		\aket{s_a,0}
	$
	and so the input is accepted with probability 1. 
		
	On the other hand, if the input is not a member, then the final affine state will have some non-zero coefficients as the values of some configuration like
	$
		\aket{s_e,l} \mbox{ and } \aket{s'_e,l}
	$
	for some $ l $. As described above, the absolute values of these non-zero coefficients are at least $ k $. Thus, any non-member will be accepted with probability at most $ \frac{1}{2k+1} $.    
	By picking a sufficiency big $ k $, the success probability can be arbitrarily close to 1.
\end{proof}

In the algorithm given in the proof, the status of the counter is never checked and for each member the value of the counter is set to zero. Thus, it is indeed a blind counter algorithm (\cite{Gre78}): The status of the counter is never checked during the computation and the input is accepted only if the value of the counter is zero at the end of computation. If the value of the counter is non-zero, the input is automatically rejected regardless of the state.

\section{Concluding remarks}

We introduced affine counter automata as an extended model of affine finite automata, and showed a separation result between exact affine and deterministic models. We also showed that a certain promise problem, which cannot be solved by bounded-error 2PFAs with sublogarithmic space and is also conjectured not to be solved by two-way quantum finite automata in polynomial time, can be solved by Las Vegas affine finite automata in linear time. Lastly, we showed that a counter helps for AfAs by showing that $ \manytwins $, which is conjectured not to be recognized by affine, quantum or classical finite state models in polynomial time, can be recognized by affine counter automata with one-sided bounded-error in realtime read mode. Since AfCAs are quantum like computation models that can use negative values, we believe that AfCAs can well characterize quantum counter automata and it remains as a future work.

\section*{Acknowledgements}

Nakanishi was supported by JSPS KAKENHI Grant Numbers 24500003, 24106009 and 16K00007, and also by the Asahi Glass Foundation. Khadiev, Vihrovs, and Yakary{\i}lmaz were supported by ERC Advanced Grant MQC. Prusis was supported by the Latvian State Research Programme NeXIT project No. 1.

\bibliographystyle{eptcs}
\bibliography{tcs}

\end{document}